%% file: entropy-arxiv-layout.tex
\documentclass[12pt, final]{article}

\usepackage{virgil_env} 
\usepackage{fullpage}
\usepackage{nips10submit_e}
\nipsfinalcopy

\usepackage{authblk}

\author[1]{Virgil Griffith}
\author[2]{Edwin K. P. Chong}
\author[3]{Ryan G. James}
\author[4]{Christopher J. Ellison}
\author[5,6]{James P. Crutchfield}
\affil[1]{Computation and Neural Systems, Caltech, Pasadena, CA 91125\vspace{.1in}}
\affil[2]{Dept.\ of Electrical \& Computer Engineering\authorcr Colorado
State University, Fort Collins, CO 80523 \vspace{.1in}}
\affil[3]{Computer Science Department, University of Colorado, Boulder, CO 80309\vspace{.1in}}
\affil[4]{Center for Complexity and Collective Computation\authorcr University of Wisconsin-Madison, Madison WI 53706\vspace{.1in}}
\affil[5]{Complexity Sciences Center and Physics Dept.\authorcr University of California Davis, Davis, CA 95616\vspace{.1in}}
\affil[6]{Santa Fe Institute, Santa Fe, NM 87501}

\usepackage{float}				
\usepackage{subfig}			

\usepackage{mathenv}
\usepackage{wasysym}

\usepackage{amssymb}
\usepackage{soul}
\usepackage{virgil_env}

\usepackage{subfloat}

\input{this_symbols.tex} 

\newtheorem{lem}{Lemma}

\newcommand*{\Icupe}[2]{\ensuremath{\opname{I}_{\cup}\!\left( #1 \! : \! #2 \right)}}

\newcommand{\Mzero}{$\mathbf{\left(M_0\right)}$\xspace}
\newcommand{\Szero}{$\mathbf{\left(S_0\right)}$\xspace}
\newcommand{\LPzero}{$\mathbf{\left(LP_0\right)}$\xspace}
\newcommand{\Iminb}{\ensuremath{\opI_\mathrm{min}}}
\newcommand{\Ired}{\ensuremath{\opI_\mathrm{red}}}

\newcommand{\Id}{$\mathbf{\left(Id\right)}$\xspace}
\newcommand{\Mone}{$\mathbf{\left(M_1\right)}$\xspace}
\newcommand{\Sone}{$\mathbf{\left(S_1\right)}$\xspace}
\newcommand{\LPone}{$\mathbf{\left(LP_1\right)}$\xspace}

\renewcommand*{\vee}{\curlyvee}
\renewcommand*{\wedge}{\curlywedge}

\title{Intersection Information Based on Common Randomness}

\begin{document}

\maketitle

\begin{abstract}
The introduction of the partial information decomposition generated a
flurry of proposals for defining an intersection information that
quantifies how much of ``the same information'' two or more random variables
specify about a target random variable. As of yet, none is wholly satisfactory. A palatable measure of intersection information would provide a principled way to quantify slippery concepts, such as synergy. Here, we introduce an intersection information measure based on the Gács-Körner common random variable that is the first to satisfy the coveted target monotonicity property. Our measure is imperfect, too, and we suggest directions for improvement.
\end{abstract}


\section{Introduction}

Partial information decomposition (PID) \cite{plw-10} is
an immensely suggestive framework for deepening our understanding of multivariate
interactions, particularly our understanding of informational redundancy and
synergy. In general, one seeks a decomposition of the mutual information
that $n$ predictors $X_1, \ldots, X_n$ convey about a target random variable,
$Y$. The intersection information is a function that calculates the
information that every predictor conveys about the target random
variable; the name draws an analogy with intersections in set theory.
An anti-chain lattice of redundant, unique and synergistic partial
information is then built from the intersection information.

As an intersection information measure,~\cite{plw-10} proposes the quantity:
\begin{equation}
\begin{split}
 \opI_{\min}\left( X_1, \ldots, X_n : Y \right)
 &= \sum_{y \in \mathcal{Y}} \Prob{y} \min_{i \in \{1, \ldots, n\}} \DKL{ \Prob{X_i|y} }{ \Prob{X_i} } \; ,
\end{split}
\label{eq:Imin_DKL}
\end{equation}
where $\opname{D_{KL}}$ is the Kullback--Leibler divergence. Although
$\opI_{\min}$ is a plausible choice for the intersection information, it has
several counterintuitive properties that make it unappealing~\cite{qsmi}. In
particular, $\opI_{\min}$ is not sensitive to the possibility that differing
predictors, $X_i$ and $X_j$, can reduce uncertainty about $Y$ in nonequivalent
ways. Moreover, the $\min$ operator effectively treats all uncertainty
reductions as the same, causing it to overestimate the ideal intersection
information. The search for an improved intersection information measure
ensued and continued through~\cite{polani12,bertschinger12,lizier13}, and today, a
widely accepted intersection information measure remains undiscovered.

Here, we do not definitively solve this problem, but explore a candidate
intersection information based on the so-called common random variable
\cite{gacs73}. Whereas Shannon mutual information is relevant to communication
channels with arbitrarily small error, the entropy of the common random
variable (also known as the zero-error information) is relevant to
communication channels without error~\cite{wolf04}. We begin by proposing a
measure of intersection information for the simpler zero-error information
case. This is useful in and of itself, because it provides a template
for exploring intersection information measures. Then, we modify our proposal,
adapting it to the Shannon mutual information case.

The next section introduces several definitions, some notation and a
necessary lemma. We extend and clarify the desired properties for
intersection information. Section~\ref{sect:Izero} introduces zero-error
information and its intersection information measure.
Section~\ref{sect:Icrv} uses the same methodology to produce a novel candidate
for the Shannon intersection information. Section~\ref{sect:examples}
shows the successes and shortcoming of our candidate intersection information
measure using example circuits and diagnoses the shortcoming's origin.
Section~\ref{sect:negsynergy} discusses the negative values of the resulting
synergy measure and identifies its origin. Section~\ref{sect:conclusion}~
summarizes our progress towards the ideal intersection information measure
and suggests directions for improvement. Appendices are devoted to technical
lemmas and their proofs, to which we refer in the main text.

\section{Preliminaries}
\subsection{Informational Partial Order and Equivalence}

We assume an underlying probability space on which we define random variables denoted by capital letters (e.g., $X$, $Y$ and $Z$). In this paper, we consider only random variables taking values on finite spaces.
Given random variables $X$ and $Y$, we write $X\preceq Y$ to signify that there exists a measurable function, $f$, such that $X=f(Y)$ almost surely ({\em i.e.}, with probability one). In this case, following the terminology in \cite{li11}, we say that $X$ is informationally poorer than $Y$; this induces a partial order on the set of \linebreak random variables.
Similarly, we write $X \succeq Y$ if $Y \preceq X$, in which case we say $X$ is
informationally richer than $Y$.

If $X$ and $Y$ are such that $X\preceq Y$ and $X\succeq Y$, then we write $X\cong Y$.
In this case, again following \cite{li11}, we say that $X$ and $Y$ are informationally equivalent.
In other words, $X \cong Y$ if and only if one can relabel the values of $X$ to obtain a random value that is equal to $Y$ almost surely and \textit{vice versa}.

This ``information-equivalence'' can easily be shown to be an equivalence
relation, and it partitions the set of all random variables into disjoint equivalence classes. The $\preceq$ ordering is invariant within these equivalence classes in the following sense. If $X\preceq Y$ and $Y\cong Z$, then $X\preceq Z$. Similarly, if $X\preceq Y$ and $X\cong Z$, then $Z\preceq Y$. Moreover, within each equivalence class, the entropy is invariant, as shown in~Section~\ref{sec:infolattice}.

\subsection{Information Lattice}
\label{sec:infolattice}

Next, we follow \cite{li11} and consider the join and meet
operators. These operators were defined for
information elements, which are $\sigma$-algebras or,
equivalently, equivalence classes of random variables. We deviate from
\cite{li11} slightly and define the join and meet operators for random
variables.

Given random variables $X$ and $Y$, we define $X\vee Y$ (called the
join of $X$ and $Y$) to be an informationally poorest
(``smallest'' in the sense of the partial order $\preceq$) random variable,
such that $X \preceq X \vee Y$ and $Y \preceq X \vee Y$. In other words, if
$Z$ is such that $X \preceq Z$ and $Y \preceq Z$, then $X \vee Y \preceq Z$.
Note that $X\vee Y$ is unique only up to equivalence with respect to $\cong$.
In other words, $X\vee Y$ does not define a specific, unique random variable.
Nonetheless, standard information-theoretic quantities are invariant over the
set of random variables satisfying the condition specified above. For example,
the entropy of $X\vee Y$ is invariant over the entire equivalence class of
random variables satisfying the condition above. Similarly, the inequality
$Z\preceq X\vee Y$ does not depend on the specific random variable chosen, as
long as it satisfies the condition above. Note, the pair $(X,Y)$ is an
instance of $X\vee Y$.

In a similar vein, given random variables $X$ and $Y$, we define $X\wedge Y$ (called the meet of $X$ and $Y$) to be an informationally richest random variable (``largest'' in the sense of $\succeq$), such that $X\wedge Y\preceq X$ and $X\wedge Y\preceq Y$. In other words, if $Z$ is such that $Z\preceq X$ and $Z\preceq Y$, then $Z\preceq X\wedge Y$. Following \cite{gacs73}, we also call $X\wedge Y$ the common random variable of $X$ and $Y$.

An algorithm for computing an instance of the common random variable between two random variables is provided in \cite{wolf04}; it generalizes straightforwardly to $n$ random variables. One can also take intersections of the $\sigma$-algebras generated by the random variables that define the meet.

The $\vee$ and $\wedge$ operators satisfy the algebraic properties of a lattice \cite{li11}. In particular, the following hold:
\begin{itemize}
\item commutative laws: $X\vee Y \cong Y\vee X$ and
	$X\wedge Y \cong Y\wedge X$;
\item associative laws: $X\vee (Y\vee Z) \cong (X\vee Y)\vee Z$ and
	$X\wedge (Y\wedge Z) \cong (X\wedge Y)\wedge Z$;
\item absorption laws: $X\vee(X\wedge Y) \cong X$ and
	$X\wedge(X\vee Y) \cong X$;
\item idempotent laws: $X\vee X \cong X$ and $X\wedge X \cong X$;
\item generalized absorption laws: if $X\preceq Y$, then $X\vee Y\cong Y$
	and $X\wedge Y \cong X$.
\end{itemize}
Finally, the partial order $\preceq$ is preserved under $\vee$ and $\wedge$, {\em i.e.}, if $X\preceq Y$, then $X\vee Z\preceq Y\vee Z$ and $X\wedge Z\preceq Y\wedge Z$.

Let $\ent{\cdot}$ represent the entropy function and $\ent{\cdot|\cdot}$ the
conditional entropy. We denote the Shannon mutual information between $X$ and
$Y$ by $\info{X}{Y}$. The following results highlight the invariance and
monotonicity of the entropy and conditional entropy functions with respect to
$\cong$ and $\preceq$~\cite{li11}. Given that $X \preceq Y$ if and only if $X =
f(Y)$, these results are familiar in information theory, but are restated here
using the current notation:
\begin{itemize}
\item[(a)] If $X\cong Y$, then $\ent{X}=\ent{Y}$,
$\ent{X\middle|Z}=\ent{Y\middle|Z}$, and $\ent{Z\middle|X} = \ent{Z\middle|Y}$.
\item[(b)] If $X\preceq Y$, then $\ent{X}\leq\ent{Y}$,
$\ent{X\middle|Z}\leq\ent{Y\middle|Z}$, and $\ent{Z\middle|X}\geq\ent{Z\middle|Y}$.
\item[(c)] $X\preceq Y$ if and only if $\ent{X\middle|Y}=0$.
\end{itemize}
\subsection{Desired Properties of Intersection Information}
\label{subsec:prop}

We denote $\Infor{X}{Y}$ as a nonnegative measure of information between
$X$ and $Y$.
For example, $\mathcal{I}$ could be the Shannon mutual information;
{\em i.e.}, $\Infor{X}{Y} \equiv \info{X}{Y}$. Alternatively, we could take
$\mathcal{I}$ to be the zero-error information. Yet, other possibilities
include the Wyner common information \cite{wyner75} or the quantum mutual
information \cite{adami97}. Generally, though, we require that
$\Infor{X}{Y} = 0$ if $Y$ is a constant, which is satisfied by both the
zero-error and Shannon information.

For a given choice of $\mathcal{I}$, we seek a function that captures the
amount of information about $Y$ that is captured by each of the predictors
$X_1, \ldots, X_n$. We say that $\Icap$ is an intersection information for
$\mathcal{I}$ if $\Icape{X}{Y} = \Infor{X}{Y}$. There are
currently 11 intuitive properties that we wish the ideal intersection
information measure, $\Icap$, to satisfy. Some are new (e.g., lower bound \LB, strong monotonicity \Mone, and equivalence-class invariance \Eq),
but most were introduced earlier, in various forms,
in~\cite{plw-10,qsmi,polani12,bertschinger12,lizier13}. They are as follows:

\begin{itemize}
\item\GP Global positivity: $\Icape{X_1,\ldots,X_n}{Y} \geq 0$.
\item\Eq Equivalence-class invariance:
$\Icape{X_1,\ldots,X_n}{Y}$ is invariant under substitution of $X_i$ (for any $i=1,\ldots,n$) or $Y$ by an informationally equivalent random variable.
 \item\TM Target monotonicity: If $Y\preceq Z$, then $\Icape{X_1,\ldots,X_n}{Y} \leq \Icape{X_1,\ldots,X_n}{Z}$.
 \item\Mzero Weak monotonicity: $\Icape{X_1,\ldots,X_n, W}{Y} \leq \Icape{X_1,\ldots,X_n}{Y}$ with equality if there exists a $Z \in \{X_1, \ldots, X_n\}$ such that $Z \preceq W$.

 \item\Szero Weak symmetry:
$\Icape{X_1,\ldots,X_n}{Y}$ is invariant under reordering of $X_1, \ldots, X_n$.
\end{itemize}

The next set of properties relate the intersection information to the chosen
measure of information between $X$ and $Y$.

\begin{itemize}
 \item\LB Lower bound: If $Q\preceq X_i$ for all $i=1,\ldots,n$, then $\Icape{X_1,\ldots,X_n}{Y}\geq \Infor{Q}{Y}$. Note that $X_1 \wedge \cdots \wedge X_n$ is a valid choice for $Q$. Furthermore, given that we require $\Icape{X}{Y} = \Infor{X}{Y}$, it follows that \Mzero implies \LB.

\item\Id Identity: $\Icape{X, Y}{X \vee Y} = \mathcal{I}(X : Y)$.

\item\LPzero Weak local positivity: For $n=2$ predictors, the derived ``partial information'' defined in \cite{plw-10} and described in Section~\ref{sect:examples} are nonnegative. If both \GP and \Mzero are satisfied, as well as
$\Icape{X_1,X_2}{Y} \geq \Infor{X_1}{Y} + \Infor{X_2}{Y} - \Infor{X_1\vee X_2}{Y}$,
then \LPzero is satisfied.

\end{itemize}

Finally, we have the ``strong'' properties:
\begin{itemize}
 \item\Mone Strong monotonicity:
$\Icape{X_1,\ldots,X_n, W}{Y} \leq \Icape{X_1,\ldots,X_n}{Y}$ with equality if there exists $Z \in \{X_1, \ldots, X_{n}, Y\}$ such that $Z \preceq W$.

 \item\Sone Strong symmetry:
$\Icape{X_1,\ldots,X_n}{Y}$ is invariant under reordering of $X_1, \ldots, X_n, Y$.

 \item\LPone Strong local positivity: For all $n$, the derived ``partial information'' defined in \cite{plw-10} is~nonnegative.

\end{itemize}

Properties \LB, \Mone and \Eq are introduced for the first time here.
However, \Eq is satisfied by most information-theoretic quantities and is
implicitly assumed by others. Though absent from our list, it is worthwhile
to also consider continuity and chain rule properties, in
analogy with the mutual~information~\cite{cover-thomas-91,bertschinger12}.

\section{Candidate Intersection Information for Zero-Error Information}
\label{sect:Izero}
\subsection{Zero-Error Information}

Introduced in \cite{wolf04}, the zero-error information, or G\'acs--K\"{o}rner common information, is a stricter variant of Shannon mutual information. Whereas the mutual information, $\info{A}{B}$, quantifies the magnitude of information $A$ conveys about $B$ with an arbitrarily small error $\epsilon > 0$, the zero-error information, denoted $\infozero{A}{B}$, quantifies the magnitude of information $A$ conveys about $B$ with exactly zero error, {\em i.e.}, $\epsilon = 0$. The zero-error information between $A$ and $B$ equals the entropy of the common random variable $A \wedge B$,
\[
 \infozero{A}{B} \equiv \ent{ A \wedge B }.
\]

Zero-error information has several notable properties, but the most salient is that it is nonnegative and bounded by the mutual information,
\[
0 \leq \infozero{A}{B} \leq \info{A}{B}.
\]

\subsection{Intersection Information for Zero-Error Information}

For the zero-error information case ({\em i.e.}, $\mathcal{I}=\opname{I}^{0}$), we
propose the zero-error intersection information
$\Iw^0\!\left( X_1,\ldots,X_n \: Y \right)$ as the maximum zero-error
information, $\infozero{Q}{Y}$, that a random variable, $Q$, conveys about $Y$,
subject to $Q$ being a function of each predictor $X_1,\ldots,X_n$:
\begin{equation}
\label{eq:Icapzero}
\begin{split}
 \Iw^0\!\left(X_1, \ldots, X_n\!:\!Y\right) &\equiv \max_{ \Pr(Q|Y) } \infozero{Q}{Y} \\
 & \hspace{0.2in} \textnormal{subject to } Q \preceq X_i \ \forall i \in \{1, \ldots, n\} \; .
\end{split}
\end{equation}

In Lemma~\ref{lem:Iw0} of Appendix~\ref{appendix:miscproofs}, it is shown
that the common random variable across all predictors is the maximizing $Q$.
This simplifies Equation~\eqref{eq:Icapzero} to:
\begin{equation}
\label{eq:zerocap}
\Iw^0\!\left(X_1, \ldots, X_n\!:\!Y\right)
 = \infozero{ X_1 \wedge \cdots \wedge X_n }{Y}
 = \ent{ X_1 \wedge \cdots \wedge X_n \wedge Y } \; .
\end{equation}

Most importantly, the zero-error information
$\Iw^0\!\left(X_1,\ldots,X_n \: Y \right)$ satisfies
nine of the 11 desired properties from Section~\ref{subsec:prop},
leaving only \LPzero and \LPone unsatisfied. See
Lemmas~\ref{lem:Iw0prop0}, \ref{lem:Iw0propInfor}, and \ref{lem:Iw0prop1}
in Appendix~\ref{app:Iw0} for details.

\section{Candidate Intersection Information for Shannon Information}
\label{sect:Icrv}

In the last section, we defined an intersection information for zero-error
information that satisfies the vast majority of the desired properties. This is a solid start, but an intersection information for Shannon mutual information remains the goal. Towards this end, we use the same method as in Equation~\eqref{eq:Icapzero}, leading to $\Iw$, our candidate intersection information measure for Shannon mutual information:
\begin{equation}
\label{eq:crvcap}
\begin{split}
 \Iw(X_1, \ldots, X_n:Y) &\equiv \max_{ \Pr(Q|Y) } \info{Q}{Y} \\
 &\hspace{0.2in} \textnormal{subject to } Q \preceq X_i \ \forall i \in \{1, \ldots, n\} \; .
\end{split}
\end{equation}

In Lemma~\ref{lem:Iw} of Appendix~\ref{appendix:miscproofs}, it is shown that
Equation~\eqref{eq:crvcap} simplifies to:
\begin{equation}
 \Iw(X_1, \ldots, X_n:Y) = \info{ X_1 \wedge \cdots \wedge X_n }{Y} \; .
\end{equation}

Unfortunately, $\Iw$ does not satisfy as many of the desired properties as
$\Iw^0$. However, our candidate, $\Iw$, still satisfies seven of the 11
properties; most importantly, the coveted \TM that, until now, had not been
satisfied by any proposed measure. See Lemmas~\ref{lem:Iwprop0},
\ref{lem:IwpropInfor} and \ref{lem:Iwprop1} in Appendix~\ref{app:Iw} for
details. Table~\ref{tbl:sometable} lists the desired properties satisfied by $\opI_{\min}$, $\Iw$ and $\Iw^0$. For reference, we also include $\Ired$, the proposed measure
from~\cite{polani12}.

\begin{table}[H]
\small
\centering
\begin{tabular}{ l l l l l l }
\toprule
\multicolumn{2}{c}{\textbf{Property}}    & \textbf{\Iminb} & \textbf{\Ired} & \textbf{$\Iw$} & \textbf{$\Iw^0$ } \\
\midrule
 \GP & Global Positivity  & \checkmark & \checkmark & \checkmark & \checkmark \\
\addlinespace
 \Eq & Equivalence-Class Invariance & \checkmark & \checkmark & \checkmark & \checkmark \\
\addlinespace
 \TM & Target Monotonicity  &  &  & \checkmark & \checkmark \\
\addlinespace
 \Mzero & Weak Monotonicity  & \checkmark &  & \checkmark & \checkmark \\
\addlinespace
 \Szero & Weak Symmetry   & \checkmark & \checkmark & \checkmark & \checkmark \\
\addlinespace
 \LB & Lower bound    & \checkmark & \checkmark & \checkmark & \checkmark \\
\addlinespace
 \Id & Identity    &  & \checkmark &  & \checkmark \\
\addlinespace
 \LPzero & Weak Local Positivity & \checkmark & \checkmark &  &  \\
\addlinespace
 \Mone & Strong Monotonicity  &  &  &  & \checkmark \\
\addlinespace
 \Sone & Strong Symmetry   &  &  &  & \checkmark \\
\addlinespace
 \LPone & Strong Local Positivity & \checkmark & &  &  \\
\bottomrule
\end{tabular}
\caption{The $\Icap$ desired properties that each measure satisfies.
 (The appendices provide proofs for $\Iw$ and $\Iw^0$.)}
\label{tbl:sometable}
\end{table}

Lemma~\ref{lem:IwleqImin} in Appendix~\ref{appendix:miscproofs} allows a
comparison of the three subject intersection information measures:
\begin{equation}
 0 \leq \Iw^0\!\left( X_1,\ldots,X_n : Y\right) \leq \Iw\!\left(X_1,\ldots,X_n:Y\right) \leq \opI_{\min}\left( X_1,\ldots,X_n : Y \right) \; .
\end{equation}

Despite not satisfying \LPzero, $\Iw$ remains an important stepping-stone towards the ideal Shannon $\Icap$. First, $\Iw$ captures what is inarguably redundant information (the common random variable); this makes $\Iw$ necessarily a lower bound on any reasonable redundancy measure. Second, it is the first proposal to satisfy target monotonicity. Lastly, $\Iw$ is the first measure to reach intuitive answers in many canonical situations, while also
being generalizable to an arbitrary number of inputs.

\section{Three Examples Comparing $\opI_{\min}$ and $\Iw$}
\label{sect:examples}

Example \textsc{Unq} illustrates how $\opI_{\min}$ gives undesirable (some
claim fatally so \cite{qsmi}) decompositions of redundant and synergistic information.
Examples \textsc{Unq} and \textsc{RdnXor} illustrate $\Iw$'s successes and
example \textsc{ImperfectRdn} illustrates $\Iw$'s paramount deficiency. For
each, we give the joint distribution $\Prob{x_1, x_2, y}$, a diagram and the decomposition derived from setting $\opI_{\min}$ or $\Iw$ as the $\Icap$ measure.
At each lattice junction, the left number is the $\Icap$ value of that node, and the number in parentheses is the $\opI_{\partial}$ value (this is the same notation used in \cite{bertschinger12}). Readers unfamiliar with the $n=2$ partial information lattice should consult \cite{plw-10}, but in short, $\opI_{\partial}$ measures the magnitude of ``new'' information at this node in the lattice beyond the nodes lower in the lattice. Specifically, the mutual information between the pair, $X_1 \vee X_2$ and $Y$, decomposes into four terms:
\begin{equation*}
\mathcal{I}(X_1 \vee X_2 : Y) =
	I_\partial(X_1, X_2 : Y) + I_\partial(X_1 : Y)
	+ I_\partial(X_2 : Y) + I_\partial(X_1 \vee X_2 : Y) ~.
\end{equation*}
In order, the terms are given by the redundant information that $X_1$ and $X_2$ both provide
to $Y$, the unique information that $X_1$ provides to $Y$, the unique information
that $X_2$ provides to $Y$ and finally, the synergistic information that
$X_1$ and $X_2$ jointly convey about $Y$. Each of these quantities can be written in terms of standard mutual information and the intersection information, $\Icap$, as follows:
\begin{equation}
\begin{split}
	\opI_{\partial}( X_1 , X_2: Y ) &= \Icape{X_1,X_2}{Y} \\
	\opI_{\partial}( X_1 : Y ) &= \mathcal{I}(X_1 : Y) - \Icape{X_1,X_2}{Y} \\
	\opI_{\partial}( X_2 : Y ) &= \mathcal{I}(X_2 : Y) - \Icape{X_1,X_2}{Y} \\
    \opI_{\partial}( X_1 \vee X_2 : Y ) &= \mathcal{I}(X_1 \vee X_2: Y) - \mathcal{I}(X_1 : Y) - \mathcal{I}(X_2 : Y)+ \Icape{X_1,X_2}{Y}
\end{split}
\label{eq:partialinfos}
\end{equation}
These quantities occupy the bottom, left, right and top nodes in the lattice
diagrams, respectively. Except for \textsc{ImperfectRdn}, measures $\Iw$ and
$\Iw^0$ reach the same decomposition for all presented examples.

\subsection{Example \textsc{Unq} (\figref{fig:UNQ})}

The desired decomposition for example \textsc{Unq} is two bits of unique information; $X_1$ uniquely specifies one bit of $Y$, and $X_2$ uniquely specifies the other bit of $Y$. The chief criticism of $\opI_{\min}$ in \cite{qsmi} was that $\opI_{\min}$ calculated one bit of redundancy and one bit of synergy for \textsc{Unq} (Figure 1c). We see that unlike $\opI_{\min}$, $\Iw$ satisfyingly arrives at two bits of unique information. This is easily seen by the~inequality,
\begin{equation}
 0 \leq \Iwe{X_1, X_2}{Y} \leq \ent{ X_1 \wedge X_2 } \leq \info{X_1}{X_2} = 0 \textnormal{ bits} \; .
\end{equation}
Therefore, as $\info{X_1}{X_2}=0$, we have $\Iwe{X_1,X_2}{Y}=0$ bits leading to $\opI_{\partial}( X_1 : Y) = 1$ bit and $\opI_{\partial}( X_2:Y) = 1$ bit (Figure 1d).

\begin{figure}[H]
 \centering
 \subfloat[]{
 \label{fig:exampleUa}
 \begin{minipage}[c]{0.4\linewidth}
  \centering
  \begin{tabular}{ c | c c} \cmidrule(r){1-2}
   $X_1$ $X_2$ &$Y$ \\
   \cmidrule(r){1-2}
   \bin{a b} & \bin{ab} & \quad \nicefrac{1}{4}\\
   \bin{a B} & \bin{aB} & \quad \nicefrac{1}{4}\\
   \bin{A b} & \bin{Ab} & \quad \nicefrac{1}{4}\\
   \bin{A B} & \bin{AB} & \quad \nicefrac{1}{4}\\
   \cmidrule(r){1-2}
  \end{tabular}
 \end{minipage}
 \begin{minipage}[c]{.4\linewidth}
  \begin{align*}
   \info{X_1\vee X_2}{Y} &= 2 \\
   \info{X_1}{Y} &= 1 \\
   \info{X_2}{Y} &= 1 \\
   \addlinespace
   \Iminn{X_1,X_2}{Y} &= 1 \\
   \Iwe{X_1,X_2}{Y} &= 0
  \end{align*}
 \end{minipage}
 }\\
 \subfloat[]{%
 \label{fig:exampleUb}
 \begin{minipage}[c]{0.4\linewidth}
  \centering\vspace{.1in}
  \includegraphics[height=1.6in]{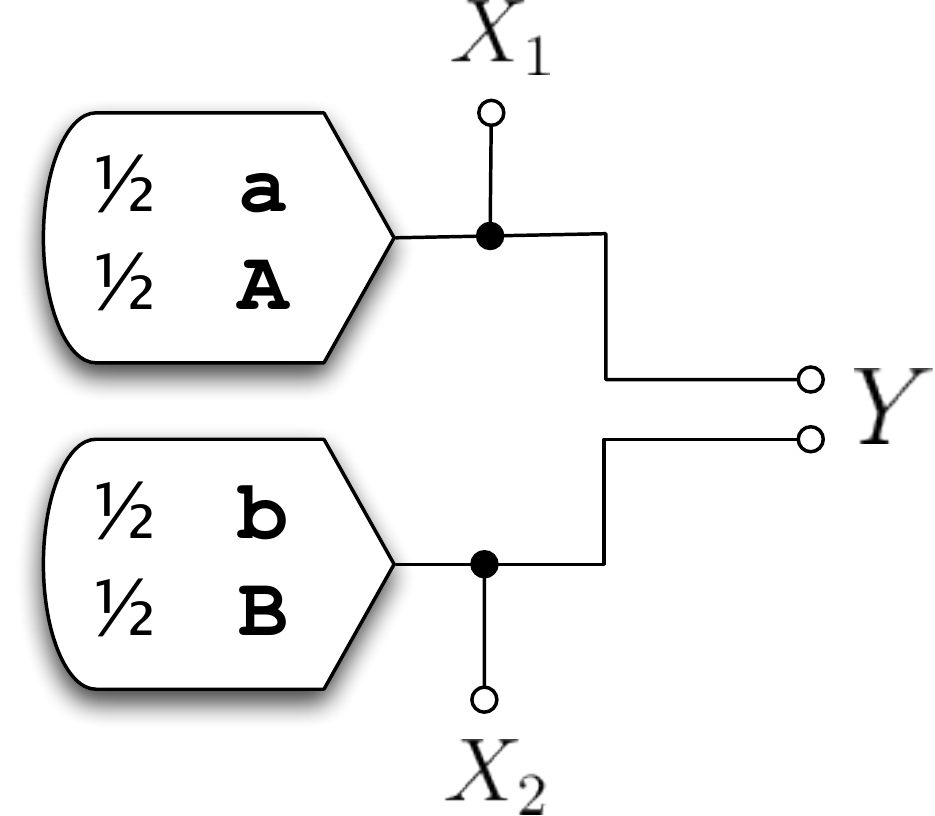}\\
  $\phantom{blah}$
 \end{minipage}
 }\\
 \subfloat[]{
 \label{fig:UNQc}
 \begin{minipage}[c]{0.4\linewidth}
  \centering
  \includegraphics[width=1.5in]{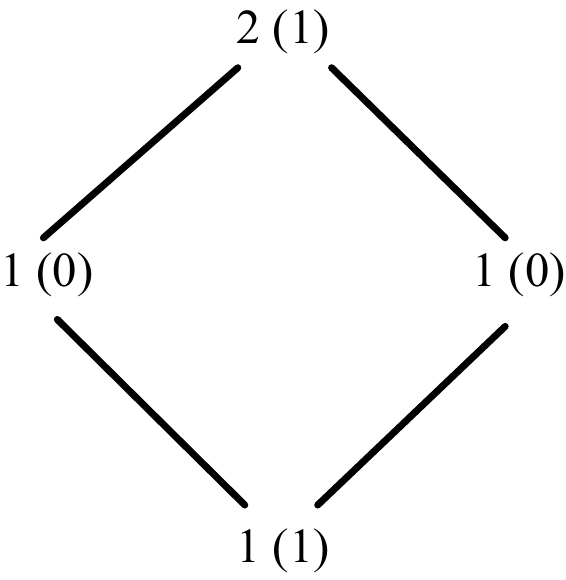}\\
  $\phantom{blah}$
 \end{minipage}
 }
 \subfloat[]{
 \label{fig:UNQd}
 \begin{minipage}[c]{0.4\linewidth}
  \centering
   \includegraphics[height=1.5in]{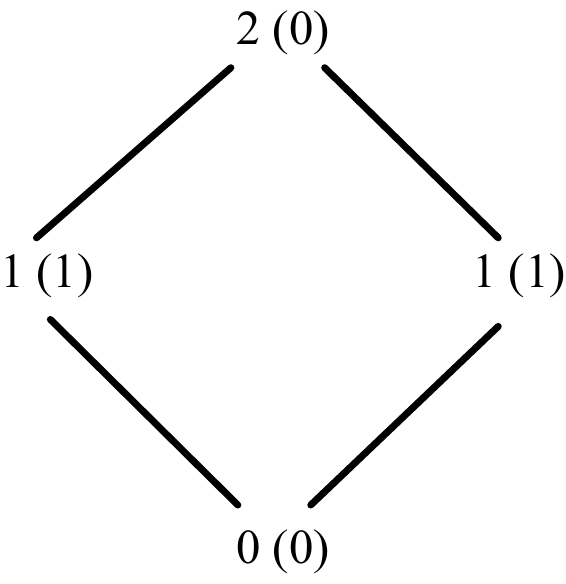}\\
  $\phantom{blah}$
 \end{minipage}
 }
 \caption{Example \textsc{Unq}. This is the canonical example of unique information. $X_1$ and $X_2$ each uniquely specify a single bit of $Y$. This is the simplest example, where $\opI_{\min}$ calculates an undesirable decomposition (c) of one bit of redundancy and one bit of synergy. $\Iw$ and $\Iw^0$ each calculate the desired decomposition (d). (\textbf{a}) Distribution and information quantities; (\textbf{b})~circuit diagram; (\textbf{c}) $\opI_{\min}$; (\textbf{d}) $\Iw$ and $\Iw^0$.}
 \label{fig:UNQ}
\end{figure}

\subsection{Example \textsc{RdnXor} (\figref{fig:RdnXor})}

In \cite{qsmi}, \textsc{RdnXor} was an example where $\opI_{\min}$ shined by reaching the desired decomposition of one bit of redundancy and one bit of synergy. We see that $\Iw$ finds this same answer. $\Iw$ extracts the common random variable within $X_1$ and
$X_2$---the \bin{r}/\bin{R} bit---and calculates the mutual information between the common random variable and $Y$ to arrive at $\Iwe{X_1,X_2}{Y}=1$ bit.

\subsection{Example \textsc{ImperfectRdn} (\figref{fig:ImperfectRdn})}

\textsc{ImperfectRdn} highlights the foremost shortcoming of $\Iw$: It does not
detect ``imperfect'' or ``lossy'' correlations between $X_1$ and $X_2$. Given
\LPzero, we can determine the desired decomposition analytically. First,
$\info{X_1\vee X_2}{Y} = \info{X_1}{Y} = 1$ bit, and thus, $\info{X_2}{Y|X_1} =
0$ bits. Since the conditional mutual information is the sum of the synergy
$\opI_{\partial}\!\left(X_1, X_2 \:Y\right)$ and unique information
$\opI_{\partial}\!\left( X_2 \:Y\right)$, both quantities must also be zero.
Then, the redundant information $\opI_{\partial}\!\left( X_1, X_2 \:Y\right) =
\info{X_2}{Y} - \opI_{\partial}( X_2 : Y ) = \info{X_2}{Y} = 0.99$ bits.
Having determined three of the partial informations, we compute the final
unique information:
$\opI_{\partial}\!\left( X_1 \:Y\right)=\info{X_1}{Y} - 0.99 = 0.01$ bits.

\begin{figure}[H]
 \centering
 \subfloat[]{
 \begin{minipage}[c]{0.4\linewidth}\centering
 \begin{tabular}{ c | c c }
  \cmidrule(r){1-2}
  $X_1$ $X_2$ & $Y$ \\
  \cmidrule(r){1-2}
  \bin{r0 r0} & \bin{r0} & \quad \nicefrac{1}{8}\\
  \bin{r0 r1} & \bin{r1} & \quad \nicefrac{1}{8}\\
  \bin{r1 r0} & \bin{r1} & \quad \nicefrac{1}{8}\\
  \bin{r1 r1} & \bin{r0} & \quad \nicefrac{1}{8}\\
  \addlinespace
  \bin{R0 R0} & \bin{R0} & \quad \nicefrac{1}{8}\\
  \bin{R0 R1} & \bin{R1} & \quad \nicefrac{1}{8}\\
  \bin{R1 R0} & \bin{R1} & \quad \nicefrac{1}{8}\\
  \bin{R1 R1} & \bin{R0} & \quad \nicefrac{1}{8}\\
  \cmidrule(r){1-2}
 \end{tabular}
 \end{minipage}
 \begin{minipage}[c]{0.4\linewidth} \centering
  \begin{align*}
   \info{X_1\vee X_2}{Y} &= 2 \\
   \info{X_1}{Y} &= 1 \\
   \info{X_2}{Y} &= 1 \\
   \addlinespace
   \Iminn{X_1,X_2}{Y} &= 1 \\
   \Iwe{X_1,X_2}{Y} &= 1
  \end{align*}
 \end{minipage}
 }\\
 \subfloat[]{
 \begin{minipage}[c]{0.3\linewidth}
  \centering\vspace{.1in}
  \includegraphics[width=2.4in]{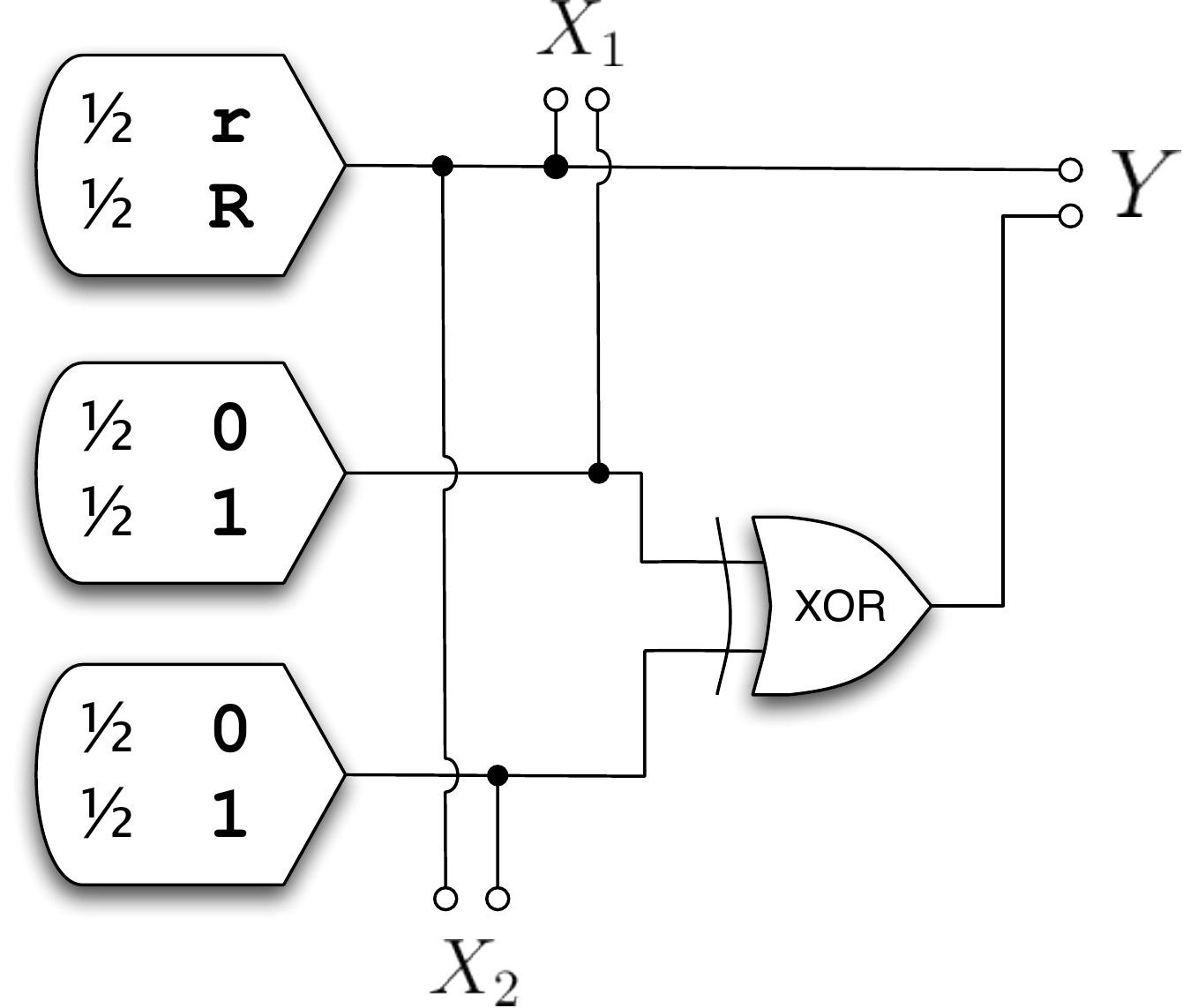}\\
  $\phantom{blah}$
 \end{minipage}
 }
 \\[1.5em]
 \subfloat[]{
 \begin{minipage}[c]{0.4\linewidth}
  \centering
  \includegraphics[width=1.5in]{PID2-RDN_XOR.pdf}\\
  $\phantom{blah}$
 \end{minipage}
 }
 \subfloat[]{
 \begin{minipage}[c]{0.4\linewidth}
  \centering
  \includegraphics[width=1.5in]{PID2-RDN_XOR.pdf}\\
  $\phantom{blah}$
 \end{minipage}
 }
 \caption{Example \textsc{RdnXor}. This is the canonical example of redundancy and synergy coexisting. $\opI_{\min}$ and $\Iw$ each reach the desired decomposition of one bit of redundancy and one bit of synergy. This is the simplest example demonstrating $\Iw$ and $\Iw^0$ correctly extracting the embedded redundant bit within $X_1$ and $X_2$. (\textbf{a}) Distribution and information quantities; (\textbf{b}) circuit diagram; (\textbf{c}) $\opI_{\min}$; (\textbf{d}) $\Iw$ and $\Iw^0$.}
 \label{fig:RdnXor}
\end{figure}

\begin{figure}[H]
 \centering
 \subfloat[]{
 \begin{minipage}[c]{0.4\linewidth}
 \centering
 \begin{tabular}{ c | c c }
  \cmidrule(r){1-2}
  $\ \, X_1 \, X_2$ &$Y$ \\
  \cmidrule(r){1-2}
  \bin{0 0} & \bin{0} & \quad $0.499$\\
  \bin{0 1} & \bin{0} & \quad $0.001$\\
  \bin{1 1} & \bin{1} & \quad $0.500$\\
  \cmidrule(r){1-2}
 \end{tabular}
 \end{minipage}
 \begin{minipage}[c]{0.4\linewidth}
 \centering
 \begin{align*}
  \info{X_1\vee X_2}{Y} &= 1 \\
  \info{X_1}{Y} &= 1 \\
  \info{X_2}{Y} &= 0.99 \\
  \addlinespace
  \Iminn{X_1,X_2}{Y} &= 0.99 \\
  \Iwe{X_1,X_2}{Y} &= 0
 \end{align*}
 \end{minipage}
 \label{fig:ANDa}
 }
 \\[1.5em]
 \subfloat[]{
 \begin{minipage}[c]{0.4\linewidth}
  \centering
  \includegraphics[width=2.3in]{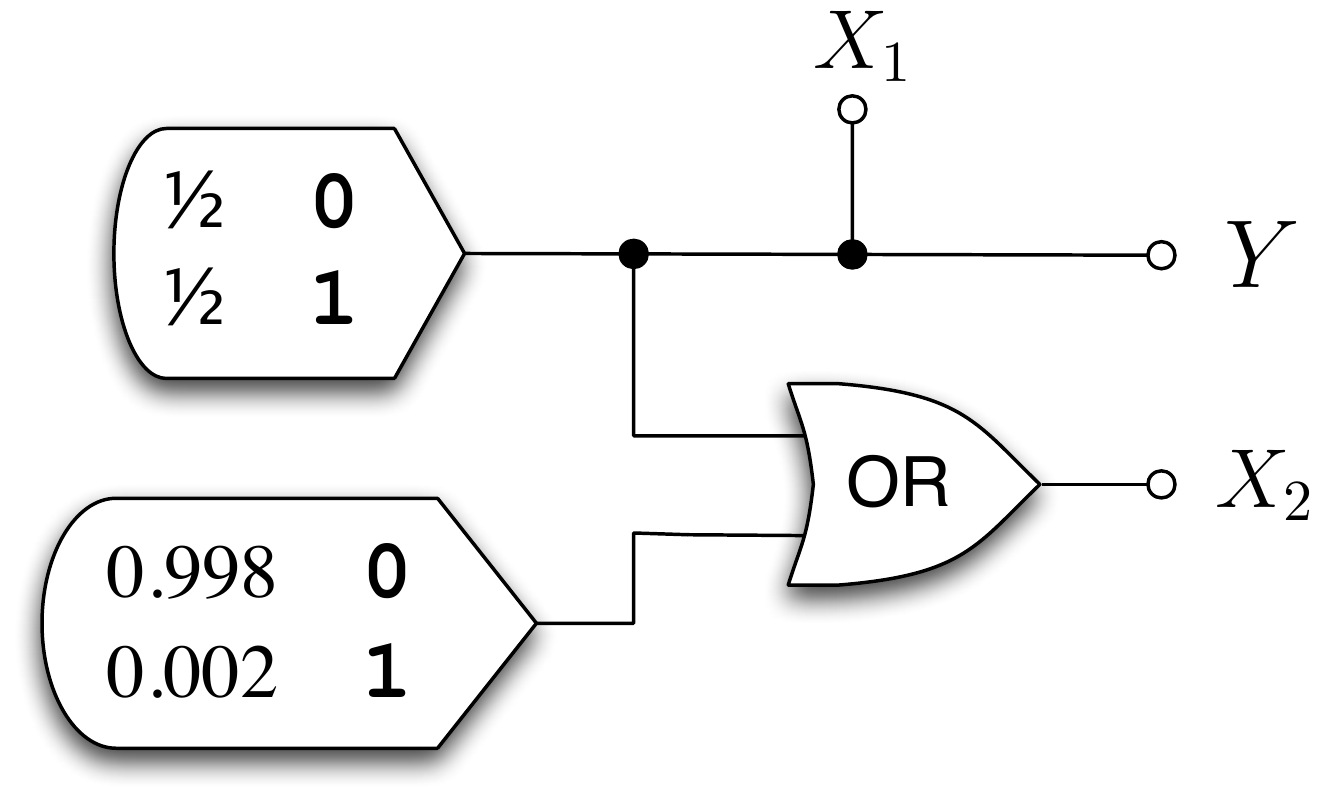}
 \end{minipage}
 }
 \\[1.5em]
 \subfloat[]{
 \begin{minipage}[c]{0.3\linewidth}
  \centering
  \includegraphics[height=1.4in]{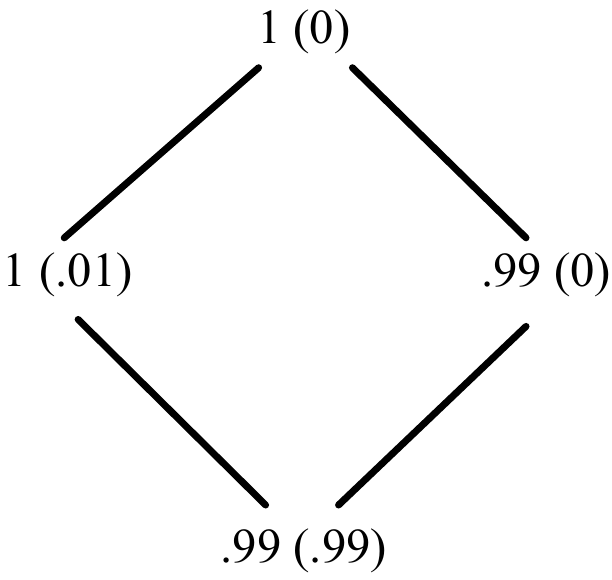}\\
  $\phantom{blah}$
 \end{minipage}
 \label{fig:ImperfectRdn_Imin}
 }
 \subfloat[]{
 \begin{minipage}[c]{0.3\linewidth}
  \centering
  \includegraphics[height=1.4in]{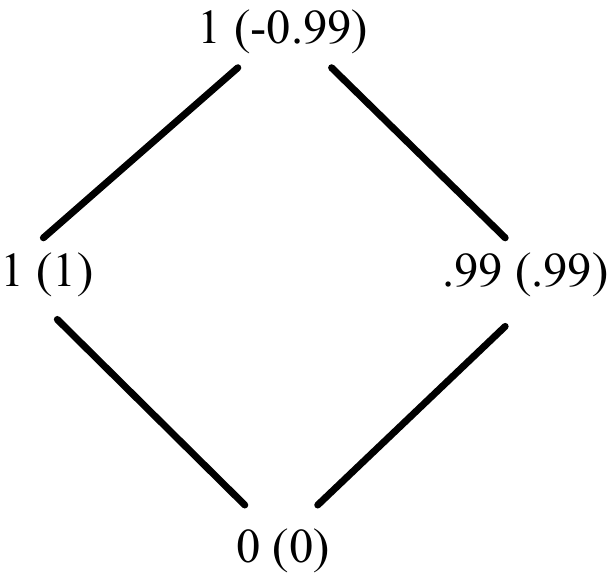}\\
  $\phantom{blah}$
 \end{minipage}
 \label{fig:ImperfectRdn_Iw}
 }
 \subfloat[]{
 \begin{minipage}[c]{0.3\linewidth}
  \centering
  \includegraphics[height=1.4in]{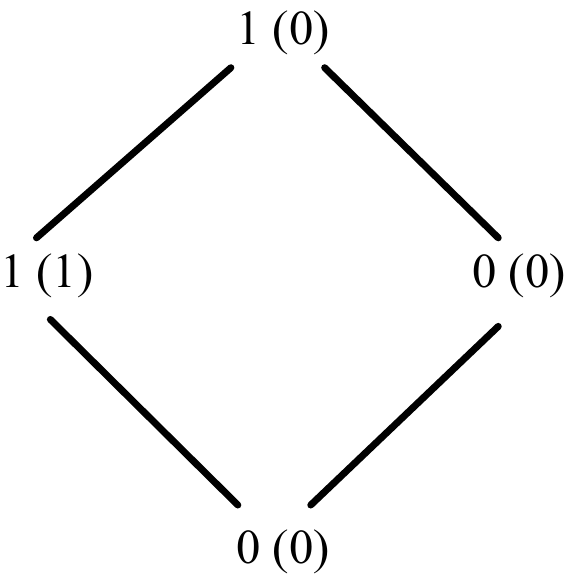}\\
  $\phantom{blah}$
 \end{minipage}
 \label{fig:ImperfectRdn_ZE}
 }
 \caption{Example \textsc{ImperfectRdn}. $\Iw$ is blind to the noisy correlation between $X_1$ and $X_2$ and calculates zero redundant information. An ideal $\Icap$ measure would detect that all \protect\linebreak of the information $X_2$ specifies about $Y$ is also specified by $X_1$ to calculate $\Icape{X_1, X_2}{Y} = 0.99$ bits. (\textbf{a}) Distribution and information quantities; (\textbf{b}) circuit diagram; (\textbf{c}) $\opI_{\min}$; (\textbf{d}) $\Iw$; \protect\linebreak(\textbf{e}) $\Iw^0$.}
 \label{fig:ImperfectRdn}
\end{figure}

How well do $\opI_{\min}$ and $\Iw$ match the desired decomposition of \textsc{ImperfectRdn}? We see that $\opI_{\min}$ calculates the desired decomposition (Figure 3c); however, $\Iw$ does not (Figure 3d). Instead, $\Iw$ calculates zero redundant information, that $\Icape{X_1,X_2}{Y}=0$ bits. This unpleasant answer arises from $\Prob{X_1 = \bin{0}, X_2 = \bin{1}} > 0$. If this were zero, then both $\Iw$ and $\opI_{\min}$ reach the desired one bit of redundant information. Due to the nature of the common random variable, $\Iw$ only sees the ``deterministic'' correlations between $X_1$ and $X_2$; add even an iota of noise between $X_1$ and $X_2$, and $\Iw$ plummets to zero. This highlights the fact that $\Iw$ is not continuous: an arbitrarily small change in the probability distribution can result in a discontinuous jump in the value of $\Iw$. As with traditional information measures, such as the entropy and the mutual information, it may be desirable to have an $\Icap$ measure that is continuous over the simplex.

To summarize, \textsc{ImperfectRdn} shows that when there are additional ``imperfect'' correlations between $A$ and $B$, {\em i.e.}, $\info{A}{B\middle|A \wedge B} > 0$, $\Iw$ sometimes underestimates the ideal $\Icape{A,B}{Y}$.

\section{Negative Synergy}
\label{sect:negsynergy}

In \textsc{ImperfectRdn}, we saw $\Iw$ calculate a synergy of $-0.99$ bits (Figure 3d). What does this mean? Could negative synergy be a ``real'' property of Shannon information? When $n=2$, it is fairly easy to diagnose the cause of negative synergy from the equation for $\opI_{\partial}(X_1 \vee X_2 : Y)$ in Equation~\eqref{eq:partialinfos}. Given \GP, negative synergy occurs if and only if,
\begin{equation}
 \info{X_1 \vee X_2}{Y} < \info{X_1}{Y} + \info{X_2}{Y} - \Icape{X_1,X_2}{Y}
 = \Icupe{X_1,X_2}{Y},
\label{eq:negsyn}
\end{equation}
where $\Icup$ is dual to $\Icap$ and related by the inclusion-exclusion principle.
For arbitrary $n$, this is $\Icup\!\left( X_1, \ldots, X_n : Y \right) \equiv \sum_{\mathbf{S} \subseteq \{X_1, \ldots, X_n\} } (-1)^{\left|\mathbf{S}\right|+1} \Icape{ S_1, \ldots, S_{|\mathbf{S}|}}{Y}$. The intuition behind $\Icup$ is that it
represents the aggregate information contributed by the sources, $X_1, \ldots, X_n$,
without considering synergies or double-counting redundancies.

From Equation~\eqref{eq:negsyn}, we see that negative synergy occurs when $\Icap$ is
small, probably too small. Equivalently, negative synergy occurs when the
joint random variable conveys less about $Y$ than the sources, $X_1$ and $X_2$,
convey separately; mathematically, when $\info{X_1 \vee X_2}{Y} <
\opI_{\cup}(X_1, X_2 : Y )$. On the face of it, this sounds strange. No
structure ``disappears'' after $X_1$ and $X_2$ are combined by the $\vee$ operator. By the definition of $\vee$, there are always functions $f_1$ and
$f_2$, such that $X_1 \cong f_1( Z )$ and $X_2 \cong f_2( Z )$. Therefore, if
your favorite $\Icap$ measure does not satisfy \LPzero, it is too strict.

This means that our measure, $\Iw^0$, does not account for the full zero-information overlap between $\infozero{X_1}{Y}$ and $\infozero{X_2}{Y}$. This is shown in the example, \textsc{Subtle} (\figref{fig:Subtle}), where $\Iw^0$ calculates a synergy of $-0.252$ bits. Defining a zero-error, $\Icap$, that satisfies \LPzero is a matter of ongoing research.

\section{Conclusions and Path Forward}
\label{sect:conclusion}

We made incremental progress on several fronts towards the ideal Shannon $\Icap$.

\subsection{Desired Properties}

We have expanded, tightened and grounded the desired properties for $\Icap$. Particularly,

\begin{itemize}
 \item \LB highlights an uncontentious, yet tighter lower bound on $\Icap$ than \GP.
 \item Inspired by $\Icape{X_1}{Y}=\Infor{X_1}{Y}$ and \Mzero synergistically implying \LB, we introduced \Mone as a desired property.
 \item What was before an implicit assumption, we introduced \Eq to better ground one's thinking.
\end{itemize}

\subsection{A New Measure}

 Based on the Gács--Körner common random variable, we introduced a new Shannon $\Icap$ measure. Our measure, $\Iw$, is theoretically principled and the first to satisfy \TM. A point to keep in mind is that our intersection information is zero whenever the distribution $\Prob{x_1, x_2, y}$ has full support; this dependence on structural zeros is inherited from the common random variable.

\begin{figure}[H]
 \centering
 \subfloat[]{
 \begin{minipage}[c]{0.4\linewidth}
  \centering
  \begin{tabular}{ c | c c }
   \cmidrule(r){1-2}
   $\ \, X_1 \, X_2$ &$Y$ \\
   \cmidrule(r){1-2}
   \bin{0 0} & \bin{00} & \quad \nicefrac{1}{3}\\
   \bin{0 1} & \bin{01} & \quad \nicefrac{1}{3}\\
   \bin{1 1} & \bin{11} & \quad \nicefrac{1}{3}\\
   \cmidrule(r){1-2}
  \end{tabular}
 \end{minipage}
 \label{fig:NEGSYN}
 \begin{minipage}[c]{0.4\linewidth}
  \centering
  \begin{align*}
   \info{X_1\vee X_2}{Y} &= 1.585 \\
   \info{X_1}{Y} &= 0.918 \\
   \info{X_2}{Y} &= 0.918 \\
   \info{X_1}{X_2} &= 0.252\\
   \addlinespace
   \Iminn{X_1,X_2}{Y} &= 0.585 \\
   \Iwe{X_1,X_2}{Y} &= 0.0
  \end{align*}
 \end{minipage}
 }
 \\[1.5em]
 \subfloat[]{
 \begin{minipage}[c]{0.4\linewidth}
  \centering
  \includegraphics[width=2.8in]{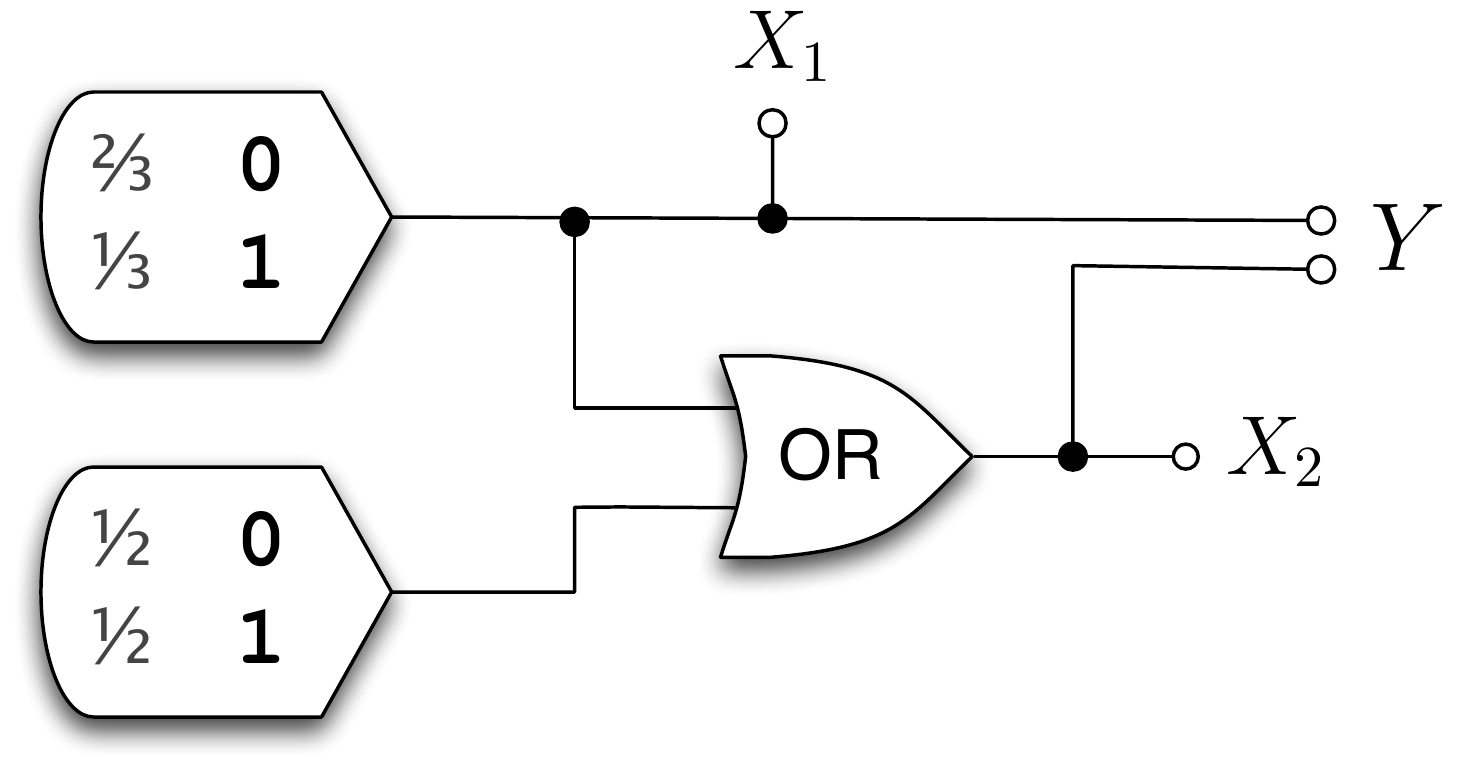}
 \end{minipage}
 \label{fig:Subtle_circuit}
 }
 \\[1.5em]
 \subfloat[]{
 \begin{minipage}[c]{0.4\linewidth}
  \centering
  \includegraphics[width=1.65in]{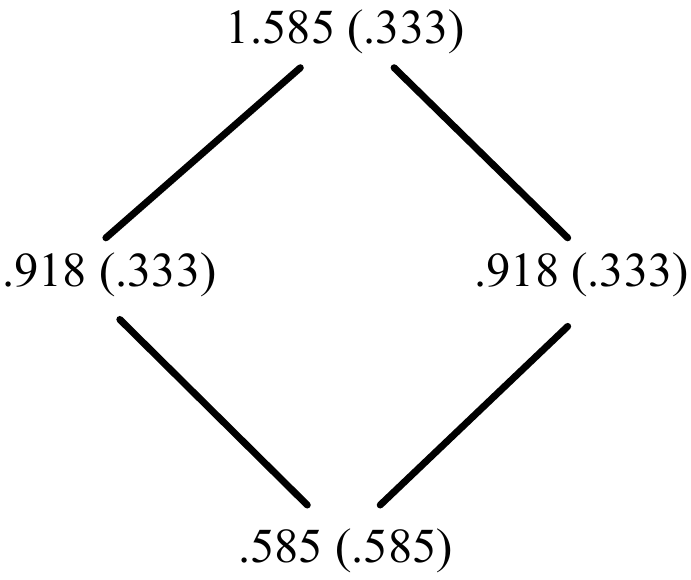}\\
  $\phantom{blah}$
 \end{minipage}
 }
 \subfloat[]{
 \begin{minipage}[c]{0.4\linewidth}
  \centering
  \includegraphics[width=1.65in]{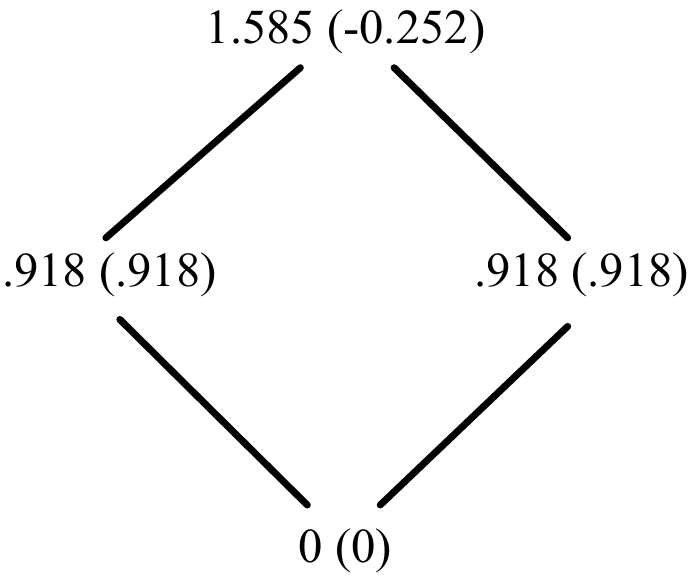}\\
  $\phantom{blah}$
 \end{minipage}
 }
 \caption{Example \textsc{Subtle}. In this example, both $\Iw$ and $\Iw^0$ calculate a synergy of $-0.252$ bits of synergy. What kind of redundancy must be captured for a nonnegative decomposition for this example? (\textbf{a}) Distribution and information quantities; (\textbf{b}) circuit diagram; (\textbf{c}) $\opI_{\min}$; \protect\linebreak (\textbf{d}) $\Iw$ and $\Iw^0$.}
 \label{fig:Subtle}
\end{figure}

\subsection{How to Improve}

 We identified where $\Iw$ fails; it does not detect ``imperfect'' correlations between $X_1$ and $X_2$. One next step is to develop a less stringent $\Icap$ measure that satisfies \LPzero for \textsc{ImperfectRdn}, while still satisfying \TM. Satisfying continuity would also be a good next step.

Contrary to our initial expectation, \textsc{Subtle}, showed that $\Iw^0$ does not satisfy \LPzero. This matches a result from~\cite{bertschinger12}, which shows that
\LPzero, \Sone, \Mzero and \Id cannot all be simultaneously satisfied, and it suggests that $\Iw^0$ is too strict. Therefore, what kind of zero-error informational overlap is $\Iw^0$ not capturing? The answer is of paramount importance. The next step is to formalize what exactly is required for a zero-error $\Icap$ to satisfy \LPzero. From \textsc{Subtle}, we can likewise see that within zero-error information, \Id and \LPzero are incompatible.

{\textbf Acknowledgments} Virgil Griffith thanks Tracey Ho, and Edwin K. P. Chong thanks Hua Li for valuable discussions. While intrepidly pulling back the veil of ignorance, Virgil Griffith was funded by a Department of Energy Computational Science Graduate Fellowship; Edwin K. P. Chong was funded by Colorado State University's Information Science \& Technology Center; Ryan G. James and James P. Crutchfield were funded by Army Research Office grant W911NF-12-1-0234; Christopher J. Ellison was funded by a subaward from the Santa Fe Institute under a grant from the John Templeton Foundation.


\section*{\noindent Appendix}
\vspace{12pt}
\appendix
By and large, most of these proofs follow directly from the lattice properties
and also from the invariance and monotonicity properties with respect to
$\cong$ and $\preceq$.

\vspace{12pt}
\noindent{\em A. Properties of $\Iw^0$}
\label{app:Iw0}
\vspace{12pt}

\begin{lem}\label{lem:Iw0prop0}
$\Iw^0\!\left(X_1,\ldots,X_n \: Y \right)$ satisfies
\GP,
\Eq,
\TM,
\Mzero, and
\Szero, but not
\LPzero.
\end{lem}

\begin{proof}
\GP follows immediately from the nonnegativity of the entropy.
\Eq follows from the invariance of entropy within the equivalence classes
induced by $\cong$. \TM follows from the monotonicity of the entropy with
respect to $\preceq$. \Mzero also follows from the monotonicity of the entropy,
but now applied to $\wedge_i X_i \wedge W \wedge Y \preceq \wedge_i X_i \wedge Y$.
If there exists some $j$, such that $X_j \preceq W$, then generalized
absorption says that $\wedge_i X_i \wedge W \wedge Y \cong \wedge_i X_i \wedge Y$,
and thus, we have the equality condition. \Szero is a consequence of the commutativity
of the $\wedge$ operator. To see that \LPzero is not satisfied by the $\Iw^0$,
we point to the example, \textsc{Subtle} (\figref{fig:Subtle}), which has negative
synergy. One can also rewrite \LPzero as the supermodularity law for
common information, which is known to be false in general.
(See \cite{li11}, Section 5.4.)
\end{proof}

\begin{lem}\label{lem:Iw0propInfor}
$\Iw^0\!\left(X_1,\ldots,X_n \: Y \right)$ satisfies
\LB,
\SR, and
\Id.
\end{lem}

\begin{proof}
For \LB, note that $Q \preceq X_1 \wedge \cdots \wedge X_n$ for any $Q$ obeying
$Q\preceq X_i$ for $i=1,\ldots,n$. Then, apply the monotonicity of the entropy.
\SR is trivially true given Lemma~\ref{lem:Iw0} and the definition of
zero-error information. Finally, \Id follows from the absorption law and
the invariance of the entropy.
\end{proof}

\begin{lem}\label{lem:Iw0prop1}
$\Iw^0\!\left(X_1,\ldots,X_n \: Y \right)$ satisfies \Mone and \Sone,
but not \LPone. \end{lem}

\begin{proof}
\Mone follows using the absorption and monotonicity of the entropy in nearly the same way that \Mzero does. \Sone follows from commutativity, and \LPone is false, because \LPzero is false.
\end{proof}

\vspace{12pt}
\noindent{\em B. Properties of $\Iw$}
\label{app:Iw}

The proofs here are nearly identical to those used for $\Iw^0$.\\

\begin{lem}\label{lem:Iwprop0}
$\Iw\!\left(X_1,\ldots,X_n \: Y \right)$ satisfies
\GP,
\Eq,
\TM,
\Mzero, and
\Szero, but not
\LPzero.
\end{lem}

\begin{proof}
\GP follows from the nonnegativity of mutual information. \Eq follows from
the invariance of entropy. \TM follows from the data processing inequality.
\Mzero follows from applying the monotonicity of the mutual information
$\info{Y}{\,\cdot\,}$
to $\wedge_i X_i \wedge W \preceq \wedge_i X_i$. If there exists some $j$, such
that $X_j \preceq W$, then generalized absorption says that
$\wedge_i X_i \wedge W \cong \wedge_i X_i$, and thus, we have the equality
condition. \Szero follows from commutativity, and a counterexample for
\LPzero is given by \textsc{ImperfectRdn} (\figref{fig:ImperfectRdn}).
\end{proof}

\begin{lem}\label{lem:IwpropInfor}
$\Iw\!\left(X_1,\ldots,X_n \: Y \right)$ satisfies
\LB and \SR, but not \Id.
\end{lem}

\begin{proof}
For \LB, note that $Q \preceq X_1 \wedge \cdots \wedge X_n$ for any $Q$ obeying
$Q\preceq X_i$ for $i=1,\ldots,n$. Then, apply the monotonicity of the mutual
information to $\info{Y}{\,\cdot\,}$. \SR is trivially true given Lemma~\ref{lem:Iw}.
Finally, \Id does not hold, since $X \wedge Y \preceq X \vee Y$, and thus,
$\Iw\!\left(X, Y \: Y \wedge Y\right) = \ent{X \wedge Y}$.
\end{proof}

\begin{lem}\label{lem:Iwprop1}
$\Iw\!\left(X_1,\ldots,X_n \: Y \right)$ does not satisfy
\Mone, \Sone, or \LPone.\end{lem}

\begin{proof}
\Mone is false due to a counterexample provided by \textsc{ImperfectRdn}
(\figref{fig:ImperfectRdn}), where $\Iwe{X_1}{Y}=0.99$ bits and
$\Iwe{X_1, Y}{Y}=0$ bits. \Sone is false, since
$\Iwe{X,X}{Y} \neq \Iwe{X,Y}{X}$. Finally, \LPone is false, due to \LPzero
being false.
\end{proof}

\vspace{12pt}
\noindent{\em C. Miscellaneous Results}
\label{appendix:miscproofs}
\vspace{12pt}

\begin{lem}\label{lem:Iw0} Simplification of $\Iw^0$.
\begin{align*}
 \Iw^0\!\left(X_1, \ldots, X_n\!:\!Y\right) &\equiv \max_{ \Pr(Q|Y) } \infozero{Q}{Y}
 \hspace{0.25in} \textnormal{subject to } Q \preceq X_i \ \forall i \in \{1, \ldots, n\} \\
 &= \ent{X_1\wedge\cdots\wedge X_n\wedge Y}
\end{align*}
\end{lem}
\begin{proof}
Recall that $\infozero{Q}{Y} \equiv \ent{Q \wedge Y}$, and note that
$\wedge_i X_i$ is a valid choice for $Q$. By definition, $\wedge_i X_i$ is the
richest possible $Q$, and so, monotonicity with respect to $\preceq$ then
guarantees that $\ent{\wedge_i X_i \wedge Y} \geq \ent{Q \wedge Y}$.
\end{proof}

\begin{lem}\label{lem:Iw} Simplification of $\Iw$.
\begin{align*}
 \Iw\!\left(X_1, \ldots, X_n\!:\!Y\right) &\equiv \max_{ \Pr(Q|Y) } \info{Q}{Y}
 \hspace{0.25in} \textnormal{subject to } Q \preceq X_i \ \forall i \in \{1, \ldots, n\} \\
 &= \info{X_1\wedge\cdots\wedge X_n}{Y}
\end{align*}
\end{lem}
\begin{proof}
Note that $\wedge_i X_i$ is a valid choice for $Q$. By definition,
$\wedge_i X_i$ is the richest possible $Q$, and so, monotonicity with respect
to $\preceq$ then guarantees that $\info{Q}{Y} \leq \info{\wedge_i X_i}{Y}$.
\end{proof}

\begin{lem}\label{lem:IwleqImin}
$\Iwe{X_1, \ldots, X_n}{Y} \leq \opI_{\min}\left( X_1, \ldots, X_n : Y \right)$
\end{lem}
\begin{proof}
We need only show that
$\opI(\wedge_i X_i : Y) \leq \opI_{\min}\left( X_1, \ldots, X_n : Y \right)$.
This can be restated in terms of the specific information:
$\opI(\wedge_i X_i : y) \leq \min_i \opI\left( X_i : y \right)$ for each $y$.
Since the specific information increases monotonically on the
lattice (\textit{cf}. Section~\ref{sec:infolattice} or~\cite{li11}), it follows that
$\opI(\wedge_i X_i : y) \leq \opI( X_j : y)$ for any $j$.
\end{proof}

\textbf{Author Contributions} Each of the authors contributed to the design, analysis, and writing of the study.

\textbf{Conflicts of Interest} The authors declare no conflicts of interest.
\bibliography{quant_synergy}




\end{document}

%% file: this_symbols.tex
\usepackage{nips10submit_e}
\nipsfinalcopy

\definecolor{mygray}{RGB}{153,153,153}
\definecolor{myred}{RGB}{221,77,57}
\definecolor{myblue}{RGB}{68,104,252}
\definecolor{mygreen}{RGB}{86,149,22}

\newcommand*{\bin}[1]{\texttt{#1}}

\newcommand*{\Ismp}{\ensuremath{\opname{I}_{\opname{smp} \xspace}}}